\documentclass[conference]{IEEEtran}
\makeatletter
\def\ps@headings{%
\def\@oddhead{\mbox{}\scriptsize\rightmark \hfil \thepage}%
\def\@evenhead{\scriptsize\thepage \hfil \leftmark\mbox{}}%
\def\@oddfoot{}%
\def\@evenfoot{}}
\makeatother
\usepackage{cite}
\usepackage{color,soul}
\usepackage[hyphens]{url}
\usepackage{graphicx}
\usepackage{enumerate}

\usepackage[utf8]{inputenc}

\let\labelindent\relax
\usepackage{enumitem}

\graphicspath{{imgs/}}
\usepackage{amsmath}
\usepackage{amsthm}
\usepackage{amssymb}
\usepackage{multirow}
\usepackage{array}
\usepackage{booktabs} 
\usepackage{algorithm}
\usepackage[noend]{algpseudocode}

\DeclareMathOperator*{\argmin}{arg\,min}

\DeclareMathOperator*{\Minimize}{Minimize:}

\DeclareMathOperator*{\SubjectTo}{Subject\phantom{a}to:}
\DeclareMathOperator*{\supp}{\textsf{supp}}

\newcommand{\norm}[1]{\left\lVert#1\right\rVert} 
\newcommand{\abs}[1]{\left\lvert#1\right\rvert}

\newtheorem{definition}{\bf{Definition}}
\newtheorem{proposition}{\bf{Proposition}}
\newtheorem{theorem}{\bf{Theorem}}

\newtheorem{remark}{Remark}
\newcolumntype{L}{>{\centering\arraybackslash}m{4cm}}

\usepackage{bm}

\begin{document}

\title{\Huge{Multi-Model Resilient Observer\\ under False Data Injection Attacks}}

\author{\IEEEauthorblockN{
\textbf{Olugbenga Moses Anubi, Charalambos Konstantinou, Carlos A. Wong, Satish Vedula}
}
\IEEEauthorblockA{Department of Electrical and Computer Engineering, FAMU-FSU College of Engineering \\
Center for Advanced Power Systems, Florida State University \\
E-mail: \{oanubi, ckonstantinou, cwong, svedula\}@fsu.edu}
}

\maketitle

\begin{abstract}
In this paper, we present the concept of boosting the resiliency of optimization-based observers for cyber-physical systems (CPS) using auxiliary sources of information. Due to the tight coupling of physics, communication and computation, a malicious agent can exploit multiple inherent vulnerabilities in order to inject stealthy signals into the measurement process. The problem setting considers the scenario in which an attacker strategically corrupts portions of the data in order to force wrong state estimates which could have catastrophic consequences. The goal of the proposed observer is to compute the true states in-spite of the adversarial corruption. In the formulation, we use a measurement prior distribution generated by the auxiliary model to refine the feasible region of a traditional compressive sensing-based regression problem. A constrained optimization-based observer is developed using \textit{$l_1$}-minimization scheme. Numerical experiments show that the solution of the resulting problem recovers the true states of the system. The developed algorithm is evaluated through a numerical simulation example of the IEEE 14-bus system.
\end{abstract}



\section{Introduction}

\IEEEPARstart{C}{yber}-physical systems (CPS) are engineered systems that are built from, and depend upon, the seamless integration of cyber and physical components. Hence, CPS are tightly integrated systems at all scales and levels that leverage information, communication, and computing systems to control a physical process in an autonomous, cooperative, intelligent, and flexible manner \cite{konstantinou2015cyber}. The decreasing cost of sensing, networking, and computation tools in the era of internet-of-things (IoT) has resulted in building complex CPS with new capabilities, reducing the cost of CPS operation, and having safer and more efficient systems. 

Many CPS applications are safety-critical systems in domains such as critical infrastructure (e.g., power grid systems), disaster monitoring, and healthcare environments. Therefore, it is of paramount importance to ensure overall stability of the physical process and avoid severe consequences. Towards that goal of maintaining normal operating conditions, a CPS is consistently monitored and controlled by data acquisition and control systems. CPS operators use measurements acquired from various sensors across the CPS infrastructure to estimate system state variables. These state estimates are critical since they are used to adjust the control of the physical space via management operations. 

In order to preserve the integrity and availability of state estimation routines in CPS-related applications, \emph{bad data detection (BDD)} mechanisms have been traditionally used to remove faulty and erroneous measurements \cite{wu2018bad}. However, recent studies have showed that judiciously falsified data can inject errors in state variables without being detected by BDD \cite{liu2011false, liang2017review, deng2017false}. Adversaries may launch such \emph{false data injection attacks (FDIAs)} able to bypass BDD functions by altering the measurements sent from the field sensing devices to the central estimation station \cite{hao2015sparse}. Furthermore, attackers may realize such FDIAs by hacking into sensors and meters or even infiltrate secondary channels of the supply chain in order to distort the measurements \cite{konstantinou2016case, konstantinou2017gps}. Fig. \ref{fig:SEattack_generic} presents a schematic of the state estimation routine under FDIA.

\begin{figure}
    \centering{\includegraphics[width=3.3in]{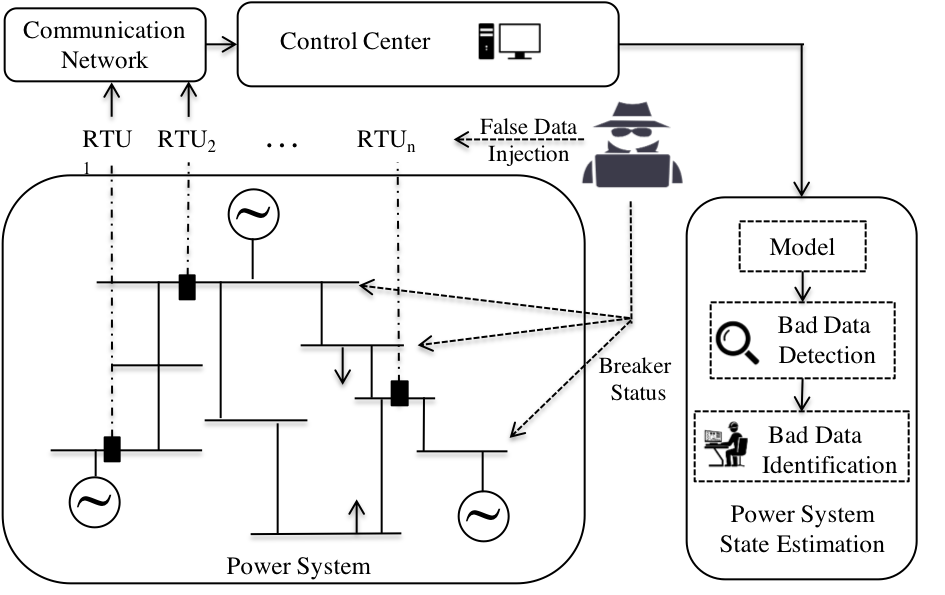}}
    \vspace{-1mm}
    \caption{State estimation under false data injection attacks (FDIAs).}
    \vspace{-2mm}
    \label{fig:SEattack_generic}
\end{figure}

Existing efforts to address the vulnerability of state estimation algorithms to FDIAs either require protection of a set of measurement sensors or verification of each state variable independently. The high computational and deployment cost, as well as significant risk involved with these approaches, have hampered their feasibility for use in practical real-time systems \cite{liang2017review, 9087789}. Furthermore, existing approaches are often developed for specific system configurations \cite{ashok2018online}. As a result, it is necessary to investigate more computationally feasible, adaptive and real-time implantable resiliency methods. 

In this work, we present an enhanced resilient state estimation approach for a dynamic CPS in which the data acquired from the sensing devices are poisoned with FDIAs. Our method relies on a data-driven model with traditional compressive sensing regression. Gaussian processes (GP) are a typical candidate for building generative probabilistic regression models from historical data \cite{anubi2019enhanced}. We demonstrate that our solution can recover the true states of the system, i.e., the system operation is able to withstand, adapt, and detect efficiently extreme adversarial FDIA settings. The developed algorithm is evaluated on a power system test case model. 

The reminder of the paper is organized as follows: in Section \ref{s:preliminaries} we provide necessary definitions and background for this work. Then, Section \ref{s:methodology} presents the formulation of the estimation problem as well as our proposed solution algorithm for the enhanced state estimator. Experimental details and simulation results are described in Section \ref{s:experiments}. Our concluding remarks are discussed in Section \ref{s:conclusions}.

\section{Notation}\label{s:notation}
The following notions and conventions are employed throughout the paper:
$\mathbb{R},\mathbb{R}^n,\mathbb{R}^{n\times m}$  denote the space of real numbers, real vectors of length $n$ and real matrices of $n$ rows and $m$ columns respectively.
$\mathbb{R}_+$ denotes positive real numbers.
$X^\top$ denotes the transpose of the quantity $X$.
By $Q\succeq0$, it is meant that $Q$ is a positive semi-definite symmetric matrix, i.e $\mathbf{x}^\top Q\mathbf{x}\ge0\hspace{1mm}\forall\mathbf{x}\neq0$ and $Q\succ0$ denotes positive definiteness which is defined with strict $>$ instead. 
Given $Q\succ0$, the $Q$-weighted norm is defined as $\norm{\mathbf{x}}_Q\triangleq\mathbf{x}^\top Q\mathbf{x}.$
Normal-face lower-case letters ($x\in\mathbb{R}$) are used to represent real scalars, bold-face lower-case letter ($\mathbf{x}\in\mathbb{R}^n$) represents vectors, while normal-face upper case ($X\in\mathbb{R}^{n\times m}$) represents matrices. Let $\mathcal{T}\subseteq\{1,\hdots,n\}$ then, for a matrix $X\in\mathbb{R}^{n\times m}$, $X_\mathcal{T} \in\mathbb{R}^{\abs{\mathcal{T}}\times m}$ is the submatrix obtained by extracting the rows of $X$ corresponding to the indices in $\mathcal{T}$. For a vector $\mathbf{x}$, $\mathbf{x}_i$ denotes its $i_{th}$ element. The support of a vector $\mathbf{x}\in\mathbb{R}^m$ is denoted by $\supp(\mathbf{x})\triangleq\left\{i\hspace{1mm}:\hspace{1mm}\mathbf{x}_i\neq0\right\}$, with $\abs{\supp(\mathbf{x})}\le m$ being the number of nonzero elements of $\mathbf{x}$. $\mathcal{S}_s^m\triangleq\left\{\mathbf{x}\in\mathbb{R}^m\}\hspace{1mm}:\hspace{1mm} 0<\abs{\supp(\mathbf{x})}\le s\right\}$ denotes the set of all nonzero $k$-sparse vectors.
Given a positive scalar $\varepsilon\in\mathbb{R}_+$, a saturation function $\textsf{sat}_\varepsilon:\mathbb{R}\mapsto[-\varepsilon,\hspace{1mm}\varepsilon]$ is given by
\begin{align*}
    \textsf{sat}_\varepsilon(x) = \left\{\begin{array}{rcl}-\varepsilon&\text{if}&x<-\varepsilon\\x&\text{if}&\abs{x}\le\varepsilon\\\varepsilon&\text{if}&x>\varepsilon\end{array}\right.
\end{align*}
A best $s_{th}$ term approximation of a vector $\mathbf{e}\in\mathbb{R}^m$ ($s\le m$) is denoted by $\mathbf{e}[s] \triangleq\min\limits_{\norm{\mathbf{f}}_0=s}\norm{\mathbf{e}-\mathbf{f}}_1$ .

\section{Preliminaries}\label{s:preliminaries}
In this section, for completeness of exposition and to facilitate faster comprehension of subsequent developments, we have gathered relevant results from literature that we built upon. 

\subsection{Overview of Resilient Estimators}
There are numerous work in literature on the secure estimation for CPS ~\cite{fawzi2014secure,hu2016secure,Fiore2017Secure, 8894484,Mishra2015Secure,Liu2017Secure, mestha2017cyber, anubi2018robust}. 
The majority of the previous research focuses on the LTI systems ranging from a Kalman filter as predictor and estimator, unconstrained $l_1$-minimization to solve the error problem, and the use of machine learning paradigms for feature discovery.
However, we focus only on the ones which are optimization based - since that is the approach we consider in this work. Moreover, due to sparsity assumption on the set of attacked nodes, majority of these works are based on the classical error correction problem \cite{candes2005decoding}. Let $\mathbb{R}^m\ni\mathbf{y} = C\mathbf{x} + \mathbf{e}$, where $C\in\mathbb{R}^{m\times n}$ is a coding matrix $(m>n)$, be a measurement vector corrupted by an arbitrary unknown but sparse error vector $\mathbf{e}$. By sparsity, we mean that $\left\|\mathbf{e}\right\|_{l_0}\le s<m$. The objective is to recover the input vector $\mathbf{x}\in\mathbb{R}^n$.
Assuming that the coding matrix $C$ is full rank, one can construct a matrix $F$ such that $FC=0$ and
\begin{align}
\tilde{\mathbf{y}} = F\mathbf{y} = F(C\mathbf{x}+\mathbf{e}) = F\mathbf{e}.
\end{align}
Thus the decoding problem is equivalent to reconstructing a sparse vector from the observation $\tilde{\mathbf{y}}=F\mathbf{e}$ and is cast as the compressive sensing problem:
\begin{align}\label{eqn:comp_sens}
\Minimize\limits_{\mathbf{e}}{\left\|\mathbf{e}\right\|_{l_0}}\hspace{2mm}\SubjectTo\hspace{2mm}\tilde{\mathbf{y}}=F\mathbf{e}.
\end{align}
Hayden et. al \cite{hayden2016sparse} obtained a sufficient condition that if all subsets of $2s$ columns of $F$ are full rank, then any error $\left\|\mathbf{e}\right\|_{l_0}\le s$ can be reconstructed uniquely by the solution of the optimization problem in \eqref{eqn:comp_sens}. 
Although in some cases~\cite{pajic2017design} the optimization problem in \eqref{eqn:comp_sens} is solved as is, in most cases, it does not lend itself to a solution in polynomial time due to its nonconvexity. As a result, it is often replaced with its convex neighbor:
\begin{align}\label{eqn:comp_sens_L1}
\Minimize\limits_{\mathbf{e}}{\left\|\mathbf{e}\right\|_{l_1}}\hspace{2mm}\SubjectTo\hspace{2mm}\tilde{\mathbf{y}}=F\mathbf{e}.
\end{align}
The two programs, however, have been shown to be equivalent under the condition that the \emph{restricted isometric property (RIP)} holds~\cite{donoho2003optimally,elad2002generalized,gribonval2003sparse,tropp2004greed}. 

\begin{definition}[RIP \cite{candes2005decoding}]
A matrix $A$ has the RIP of sparsity $k$ if there exists $0<\delta<1$ such that
\begin{align}
    \left(1-\delta\right)\norm{\mathbf{x}}_2^2\le\norm{A\mathbf{x}}_2^2\le\left(1+\delta\right)\norm{\mathbf{x}}_2^2
\end{align}
for all $\mathbf{x}\in\mathcal{S}_s$. Moreover, the smallest $\delta$ for which the above inequality holds is called the \emph{restricted isometry constant}, and denoted as $\delta_s(A)$.
\end{definition}
The above definition essentially requires that every set of columns with cardinality less that or equal to $s$ behaves like an orthonormal system. The following theorem lists the recovery error due to relaxed convex program above.

\begin{theorem}[\cite{candes2005decoding},\cite{cai2013sparse}]\label{thm:RIP}
Let $\mathbf{e}$ be a sparse vector satisfying $\tilde{\mathbf{y}}=F\mathbf{e}$ and $\hat{\mathbf{e}}$ be the solution of \eqref{eqn:comp_sens_L1}. If $\displaystyle \delta_{2s}(F)<\frac{1}{\sqrt{2}}$, then 

\begin{align}
\begin{split}
    \norm{\hat{\mathbf{e}} -\mathbf{e}}_2 &\le \\
    &\frac{2}{\sqrt{s}}\left(\frac{\delta_{2s}+\sqrt{\delta_{2s}\left(\frac{1}{\sqrt{2}}-\delta_{2s}\right)}}{\sqrt{2}\left(\frac{1}{\sqrt{2}}-\delta_{2s}\right)}+1\right) \norm{\mathbf{e}-\mathbf{e}[s]}_1,
    \end{split}
\end{align}

where $\mathbf{e}[s]$ is the best $s$-term approximation of $\mathbf{e}$.
\end{theorem}
\begin{remark}
If $\mathbf{e}\in\mathcal{S}_s$, then $\hat{\mathbf{e}}=\mathbf{e}$. Thus, if $\displaystyle \delta_{2s}(F)<\frac{1}{\sqrt{2}}$ the relaxed program in \eqref{eqn:comp_sens_L1} will recover any $s$-sparse vector $\mathbf{e}\in\mathcal{S}_s$ exactly! 
\end{remark}

Now, consider the discrete LTI system
\begin{align}
\mathbf{x}_{k+1}&=A\mathbf{x}_k\\
\mathbf{y}_k& = C\mathbf{x}_k + \mathbf{e}_k,
\end{align}
where $\mathbf{x}_k\in\mathbb{R}^n$ represents the state of the system at time $k\in\mathbb{N}$, $\mathbf{y}_k\in\mathbb{R}^m$ is the output of the monitoring nodes at time $k$ and $\mathbf{e}_k\in\mathbb{R}^m$ denote the attack signals injected by malicious agents at the monitoring nodes. Let $\mathcal{K} \subset \{1,2,\hdots,m\}$ denote the set of attacked nodes, then for all $k$, $\abs{\supp(\mathbf{e}_k)}\subset\mathcal{K}$. The resilient estimation problem is then defined as reconstructing the initial state $\mathbf{x}_0$ from corrupt measurement $\left\{\mathbf{y}_k\right\}_{k=0}^T, T\in\mathbb{N}$. We look at two scenarios from literature: $\mathcal{K}$ is time-invariant\cite{fawzi2014secure,pasqualetti2013attack} and $\mathcal{K}$ is time-varying\cite{hu2016secure}.
\subsubsection{Secure estimation for fixed attacked nodes\cite{fawzi2014secure}}
Assuming that the set $\mathcal{K}$ of attacked nodes is time-invariant:
\begin{definition}\label{def:K_invariant}
$s$ errors are correctable after $T$ steps by the decoder $\mathcal{D}:\left(\mathbb{R}^m\right)^T\mapsto\mathbb{R}^n$ if for any $\mathbf{x}_0\in\mathbb{R}^n$, any $\mathcal{K}\subset\left\{1,2,\hdots,m\right\}$ with $\abs{\mathcal{K}}\le s$, and any sequence of vectors $\mathbf{e}_0,\hdots,\mathbf{e}_{T-1}\in\mathbf{R}^m$ such that $\supp(\mathbf{e}_k)\subset\mathcal{K}$, we have $\mathcal{D}(\mathbf{y}_0,\hdots,\mathbf{y}_{T-1})=\mathbf{x}_0$, where $\mathbf{y}_k = CA^k\mathbf{x}_0 + \mathbf{e}_k$ for $k=0,1,\hdots,T-1$.
\end{definition}
\begin{proposition}
Let $T \in \mathbb{N}  \backslash \{ 0\}$. The following are equivalent:\\
(i) There is a decoder that can correct $q$ errors after $T$ steps;\\
(ii) For all $\mathbf{z}\in \mathbb{R}^n \backslash \{0\}$, $\lvert \supp(C\mathbf{z}) \cup \supp(CA\mathbf{z}) \cup \cdots \cup \supp(CA^{T-1} \mathbf{z}) \rvert > 2s$.
\end{proposition}
Consequently, the following optimal decoder is defined for when the set of attacked nodes is fixed:
\begin{align}
\mathbf{x}_0 = \argmin\limits_{\mathbf{x}} \norm { Y_T - \Phi_T (\mathbf{x}) }_{l_0}
\end{align}
where
\begin{align*}
Y_T = \left[\begin{array}{c|c|c|c} \mathbf{y}_0 & \mathbf{y}_1 & \hdots &  \mathbf{y}_{T-1} \end{array}\right] \in \mathbb{R}^{m\times T}
\end{align*}
and $\Phi_T : \mathbb{R}^n\mapsto\mathbb{R}^{m\times T}$ is a linear map given by:
\begin{align*}
\Phi_T (\mathbf{x}) = \left[\begin{array}{c|c|c|c} C\mathbf{x} & CA\mathbf{x}  & \hdots & CA^{T-1} \mathbf{x} \end{array}\right] \in \mathbb{R}^{m\times T}.
\end{align*}
\subsubsection{Secure estimation for varying attacked nodes\cite{hu2016secure}}
Assuming that the set $\mathcal{K}$ of attacked nodes can change with time but bounded as in $\abs{\mathcal{K}}\le s$:
\begin{definition}\label{def:K_variant}
$q$ errors are correctable after $T$ steps by the decoder $\mathcal{D}:\left(\mathbb{R}^m\right)^T\mapsto\mathbb{R}^n$ if for any $\mathbf{x}_0 \in \mathbb{R}^n$ and any sequence of vectors $\mathbf{e}_0,\hdots,\mathbf{e}_{T-1}\in\mathbf{R}^m$ such that $\abs{\textsf{supp}(\mathbf{e}_k)} \leq s$, we have $\mathcal{D}(\mathbf{y}_0,\hdots,\mathbf{y}_{T-1})=\mathbf{x}_0$, where $\mathbf{y}_k = CA^k\mathbf{x}_0 + \mathbf{e}_k$ for $k=0,1,\hdots,T-1$.
\end{definition}
\begin{proposition}
Let $T \in \mathbb{N}  \backslash \{ 0\}$. The following are equivalent:\\
(i) There is a decoder that can correct $s$ errors after $T$ steps;\\
(ii) For all $\mathbf{z}\in \mathbb{R}^n \backslash \{0\}$ , $\sum\limits_{k=0}^{T-1}\abs{\textsf{supp}(CA^k \mathbf{z})} > 2s$.
\end{proposition}
Consequently, the following optimal decoder is defined for when the set of attacked nodes is not fixed:
\begin{align}\label{eqn:decoder_var_attack}
\mathbf{x}_0 = \argmin\limits_{\mathbf{x}} \norm { \mathbf{y}_{(T)} - \Phi_{(T)} \mathbf{x} }_{l_0}
\end{align}
where
\begin{align*}
\mathbf{y}_{(T)} = \left[\begin{array}{c} \mathbf{y}_0\\ \mathbf{y}_1 \\ \vdots \\  \mathbf{y}_{T-1} \end{array}\right] \in \mathbb{R}^{m T},
\end{align*}
\begin{align*}
\Phi_{(T)}  = \left[\begin{array}{c} C\\ CA \\ \vdots \\ CA^{T-1} \end{array}\right] \in \mathbb{R}^{mT \times n}.
\end{align*}

\section{Resilient Observer Development}\label{s:methodology}
Consider the concurrent models

\begin{align}\label{eqn:model_based}
\begin{array}{r}
\mathbf{x}_{k+1}\\
\mathbf{y}_k
\end{array}&\begin{array}{c}=\\=\end{array}\begin{array}{l}A\mathbf{x}_k + B\mathbf{u}_k\\C\mathbf{x}_k+\mathbf{e}_k\end{array}\\\label{eqn:data_driven}
    \mathbf{y}_k &\sim\mathcal{N}(\boldsymbol{\mu}(\mathbf{z}_k),\Sigma(\mathbf{z}_k))
\end{align}

\noindent consisting of a physics-based model \eqref{eqn:model_based} and a data-driven prior \eqref{eqn:data_driven} given as a function of the auxiliary variable $\mathbf{z}\in\mathbb{R}^p$.
The data-driven model in \eqref{eqn:data_driven} gives a prior distribution on the system measurements as a function of measured auxiliary variables $\mathbf{z}_k\in\mathbb{R}^p$. This provides additional layer of security by: \textit{1)} requiring the attacker to have knowledge of the auxiliary model and the parameters, and \textit{2)} limiting the magnitude of possible state corruption. For a more detailed explanation of the advantages of the concurrent models in \eqref{eqn:model_based} and \eqref{eqn:data_driven}, as well as the resulting theoretical limits on the size of feasible attacks, interested readers are referred to the references \cite{anubi2019enhanced,anubi2019resilient} and the references therein.

Let $Y_k\triangleq\left\{\mathbf{y}_k,\mathbf{y}_{k-1},\hdots,\mathbf{y}_{k-T+1}\right\}\subset \mathbb{R}^m$ and $U_{k-1}\triangleq\left\{\mathbf{u}_{k-1},\mathbf{u}_{k-2},\hdots,\mathbf{u}_{k-T+1}\right\}\subset\mathbb{R}^l$ be collections of the last $T$-samples of the system known input and output measurements respectively. The proposed resilient observer attempts to solve the following moving horizon optimization problem for all time instant $k\ge T$:

\begin{align}\label{eqn:enh_res_est}
\begin{array}{ll}
\Minimize & \sum\limits_{i=k-T+1}^{k}{\left\|\mathbf{y}_i-C\mathbf{x}_i\right\|_{l_0}} \\
\SubjectTo&\\
        &\begin{array}{l}
          \mathbf{x}_{i+1}-A\mathbf{x}_i - B\mathbf{u}_i=0,\\\hspace{1cm}i=k-T+1,\hdots,k-1\\
          C\mathbf{x}_k\in\mathcal{Y}(\mathbf{z}_k)
        \end{array}
\end{array}
\end{align}
where the convex set $\mathcal{Y}(\mathbf{z})$ has the property that:
\begin{align}
    p(\mathbf{y_k}^*\in\mathcal{Y}|\mathbf{z}_k,\mathcal{D})\ge\tau.
\end{align}

\noindent More insight is provided in Theorem~\ref{thm:thm2}. The idea is essentially seeking historical and current state vectors, together with the minimum attacked channels, which completely explains the observations while satisfying the physics-based model and having a high likelihood according to the auxiliary model prior. The optimization parameter $\tau\in(0,\hspace{2mm}1]$ controls the likelihood threshold. It can be set to a constant value or optimized with respect to some higher-level objectives. Thus, the resilient observer optimization problem is equivalent to:

\begin{align}\label{eqn:enh_res_est2}
\begin{array}{ll}
\Minimize & \sum\limits_{i=k-T+1}^{k}{\left\|\mathbf{y}_{i}-C\mathbf{x}_{i}\right\|_{\ell_0}} \\
\SubjectTo&\\
        &\begin{array}{l}
          \mathbf{x}_{i+1}-A\mathbf{x}_{i} - B\mathbf{u}_{i}=0,\\
          \norm{C\mathbf{x}_k-\boldsymbol{\mu}(\mathbf{z}_k)}_{\Sigma^{-1}(\mathbf{z})}^2 \le \chi^2_{m}(\tau)
        \end{array}
\end{array}
\end{align}

\noindent where $\chi^2_{m}(\tau)$ is the quantile function for probability $\tau$ of the chi-squared distribution with $m$ degrees of freedom.

However, the nonconvexity due to the index minimization objective makes the optimization problem in \eqref{eqn:enh_res_est2} challenging, at best, for gradient-based solution algorithms. This will make it difficult, if not impossible, to synthesize a pragmatic algorithm that can be implemented real-time for the observer. Thus, we seek convex approximation alternatives. Fortunately, as discussed in the preliminaries Section \ref{s:preliminaries}, it is possible to approximate the index minimization objective using an $\ell_1$-norm without loosing global optimality -- provided the RIP condition holds. Consequently, the proposed resilient multi-model observer is given via the following convex program:
\begin{align}\label{eqn:enh_res_est_convex}
\begin{array}{ll}
\Minimize & \sum\limits_{i=k-T+1}^{k}{\left\|\mathbf{y}_{i}-C\mathbf{x}_{i}\right\|_{\ell_1}} \\
\SubjectTo&\\
        &\begin{array}{l}
          \mathbf{x}_{i+1}-A\mathbf{x}_{i} - B\mathbf{u}_{i}=0,\\
          \norm{C\mathbf{x}_k-\boldsymbol{\mu}(\mathbf{z}_k)}_{\Sigma^{-1}(\mathbf{z})}^2 \le \chi^2_{m}(\tau).
        \end{array}
\end{array}
\end{align}

After some algebraic manipulations and simplifications, the above program is equivalent to the following quadratically constrained basis pursuit problem: 
\begin{align}\label{eqn:enh_res_est_convex_1}
\begin{array}{ll}
\Minimize & \norm{\mathbf{y}_{(T)}-H_{(T)}\mathbf{u}_{(T-1)}-\Phi_{(T)}\mathbf{x} }_1 \\
\SubjectTo&\\
        &\norm{\Phi_T\mathbf{x}+H_T\mathbf{u}_{(T-1)}-\boldsymbol{\mu}(\mathbf{z_k})}_{\Sigma^{-1}(\mathbf{z}_k)}^2\le \chi^2_{m}(\tau),
\end{array}
\end{align}
where
\begin{align*}
\mathbf{y}_{(T)} = \left[\begin{array}{c} \mathbf{y}_{k-T+1}\\ \mathbf{y}_{k-T+2} \\ \vdots \\  \mathbf{y}_k \end{array}\right] \in \mathbb{R}^{m T},
\end{align*}
\begin{align*}
\mathbf{u}_{(T-1)} = \left[\begin{array}{c} \mathbf{u}_{k-T+1}\\ \mathbf{u}_{k-T+2} \\ \vdots \\  \mathbf{u}_{k-1} \end{array}\right] \in \mathbb{R}^{l (T-1)},
\end{align*}
where $\Phi_{(T)}$ is defined as a results of \eqref{eqn:decoder_var_attack},
\begin{align*}
H_{(T)}  = \left[\begin{array}{cccc}0&0&\hdots&0\\CB&0&\hdots&0\\CAB&CB&\hdots&0\\\vdots&\vdots&&\vdots\\CA^{T-2}B&CA^{T-3}B&\hdots&CB  \end{array}\right] \\ \in \mathbb{R}^{mT \times l(T-1)},
\end{align*}
and $\Phi_T, H_T$ are the last $m$ rows of $\Phi_{(T)}, H_{(T)}$ respectively. In the above formulation, the solution of the optimization problem gives an estimate $\hat{\mathbf{x}}_{k-T+1}$ of the state vector $\mathbf{x}_{k-T+1}$ which is then propagated forward to obtain an estimate of the current state using the physics-based dynamical model as follows
\begin{align}
    \hat{\mathbf{x}}_k = A^{T-1}\hat{\mathbf{x}}_{k-T+1} + G\mathbf{u}_{(T-1)},
\end{align}
where
\begin{align*}
    G = \left[\begin{array}{cccc}A^{T-2}B&A^{T-3}B&\hdots&B\end{array}\right].
\end{align*}

Suppose receding horizon $T$ is chosen big enough (i.e $T\ge n$) and the pair $(A,C)$ is observable, then there exists a matrix $F_{(T)}$ such that $F_{(T)}\Phi_{(T)} = 0.$\footnote{Let the singular value decomposition of $\Phi_{(T)}$ be given by
\[
    \Phi_{(T)} = \left[\begin{array}{c|c}U_1&U_2\end{array}\right]\left[\begin{array}{@{}c}\begin{smallmatrix}\sigma_1&&&\\&\sigma_2&&\\&&\ddots&\\&&&\sigma_n\end{smallmatrix}\\\hline\\\text{\huge{0}}\end{array}\right]V^\top,
\]
Then $F_{(T)} = U_2^\top$ is an example of such matrix.
}
Consequently, the optimization problem above is equivalent to
\begin{align}\label{eqn:enh_res_est_convex_2}
\begin{array}{ll}
\Minimize & \norm{ \mathbf{e} }_1 \\
\SubjectTo&\\
        &\mathbf{f}_{(T)} = F_{(T)}\mathbf{e}\\
        &\norm{\mathbf{y}_{T}+\mathbf{e}_T-\boldsymbol{\mu}(\mathbf{z_k})}_{\Sigma^{-1}(\mathbf{z}_k)}^2\le \chi^2_{m}(\tau),
\end{array}
\end{align}
where 
\begin{align*}
    \mathbf{f}_{(T)} = F_{(T)}\left(\mathbf{y}_{(T)} - H\mathbf{u}_{(T-1)}\right),
\end{align*}
and $\mathbf{e}_T, \mathbf{y}_T\in\mathbb{R}^m$ is the vector containing only the last $m$ elements of the respective vectors $\mathbf{e},\mathbf{y}$ in  order. While the form given in \eqref{eqn:enh_res_est_convex_1} is more intuitive and implemented for the simulation results, the form in \eqref{eqn:enh_res_est_convex_2} is more suitable to proof the main result which is given next. 

\begin{theorem}\label{thm:thm2}
Given a dataset $\mathcal{D} = \left\{\mathbf{Z},\mathbf{Y}\right\}$ containing historical auxiliary variables $\mathbf{Z}\in\mathbb{R}^{p\times T}$ and corresponding sensor measurements $\mathbf{Y}\in\mathbb{R}^{m\times T}$. Suppose that the latent sensor measurement satisfies the data-driven GPR prior given in \eqref{eqn:data_driven} and that there exists $\tau\in(0,\hspace{2mm}1)$ such that the true measurement $\mathbf{y}_k^*$ satisfies $p(\mathbf{y_k}^*|\mathbf{z}_k,\mathcal{D})\ge\tau$. Consider the convex optimization problem in \eqref{eqn:enh_res_est_convex_1}. Let $\hat{\mathbf{e}}$ be the solution of the equivalent form in \eqref{eqn:enh_res_est_convex_2}. If $\displaystyle \delta_{2s}(F_{(T)})<\frac{1}{\sqrt{2}}$, then, for any feasible sparse vector $\mathbf{e}$, 
\begin{align}
    \norm{\hat{\mathbf{e}}_T-\mathbf{e}_T}_2\le K_1 \textsf{sat}_{1}\left(K_2\norm{\mathbf{e}-\mathbf{e}[s]}_2\right),
\end{align}
where
\begin{align*}
    K_1 &= \sqrt{2\chi^2_{m}(\tau)\overline{\sigma}(\mathbf{z}_k)}\\
    K_2 &= K_3\sqrt{\frac{m-s}{2\chi^2_{m}(\tau)\overline{\sigma}(\mathbf{z}_k)}},
\end{align*}
with
\begin{align*}
    K_3 = \frac{2}{\sqrt{s}}\left(\frac{\delta_{2s}+\sqrt{\delta_{2s}\left(\frac{1}{\sqrt{2}}-\delta_{2s}\right)}}{\sqrt{2}\left(\frac{1}{\sqrt{2}}-\delta_{2s}\right)}+1\right)
\end{align*}
and $\overline{\sigma}(\mathbf{z}_k)$ is the biggest singular value of $\Sigma(
\mathbf{z}_k)$.
\end{theorem}
\begin{proof}
Following the development in earlier part of this section, it is straightforward to see that $p(\mathbf{y_k}^*|\mathbf{z}_k,\mathcal{D})\ge\tau$ implies that $\norm{\mathbf{y}_{T}+\mathbf{e}_T-\boldsymbol{\mu}(\mathbf{z_k})}_{\Sigma^{-1}(\mathbf{z}_k)}^2\le \chi^2_{m}(\tau)$ for any composite measurement vector $\mathbf{y}_{T}$ corrupted by the composite sparse signal $\mathbf{e}_T$. In order words, the inequality holds for all sparse signal $\mathbf{e}_T = \left[\mathbf{e}_{k-T+1}^\top\hdots\mathbf{e}_k^\top\right]^\top$, with each $\mathbf{e}_k\in\mathcal{S}_s$ satisfying $\mathbf{y}_k = \mathbf{y}_k^* + \mathbf{e}_k$. More, there exists the \emph{true state vector} $\mathbf{x}_k^*\in\mathbb{R}^n$ such that $\mathbf{y}_k^* = C\mathbf{x}_k^*$. This implies that
\begin{align*}
    \mathbf{y}_{(T)} &= \mathbf{y}_{(T)}^* + \mathbf{e}_{(T)}\\
                     &= \Phi_{(T)}\mathbf{x}_{k-T+1} + H_{(T)}\mathbf{u}_{(T-1)} + \mathbf{e}_{(T)}.
\end{align*}
From which it follows that 
\begin{align*}
    \mathbf{f}_{(T)} &= F_{(T)}\left(\mathbf{y}_{(T)} - H\mathbf{u}_{(T-1)}\right)\\
    &= F_{(T)}\mathbf{e}_{(T)}.
\end{align*}
Thus for any $\mathbf{y}_{(T)}$, the set of all $\mathbf{e}_{(T)}$ for which the quadratic inequality holds is a subset of the pre-image of $\mathbf{f}_{(T)}$ under the linear transformation $F_{(T)}$. Thus, using Lemma~\ref{thm:RIP}, the optimal point $\hat{\mathbf{e}}_{(T)}$ of the problem in \eqref{eqn:enh_res_est_convex_2} satisfies

\begin{align*}
    \norm{\hat{\mathbf{e}}_{(T)}-\mathbf{e}_{(T)}}_2&\le K_3\norm{\mathbf{e}_{(T)}-\mathbf{e}_{(T)}[s]}_1\\
    &\le K_3\sqrt{m-s}\norm{\mathbf{e}_{(T)}-\mathbf{e}_{(T)}[s]}_2.
\end{align*}
Thus,
\begin{align}\nonumber
    \norm{\hat{\mathbf{e}}_{T}-\mathbf{e}_{T}}_2&\le \norm{\mathbf{e}_{(T)}-\mathbf{e}_{(T)}}_2\\\label{eqn:proof_thm2_1}
    &\le K_3\sqrt{m-s}\norm{\mathbf{e}_{(T)}-\mathbf{e}_{(T)}[s]}_2
\end{align}
Moreover, adding and subtracting $\hat{\mathbf{e}}_T$ in the quadratic inequality constraint and using the left-hand-side triangular inequality  yields the following sequence of inequalities;
\begin{align*}
    \norm{\mathbf{y}_{T}+\hat{\mathbf{e}}_T -\boldsymbol{\mu}(\mathbf{z_k}) - \hat{\mathbf{e}}_T + \mathbf{e}_T}_{\Sigma^{-1}(\mathbf{z}_k)}^2\le \chi^2_{m}(\tau)
\end{align*}
\begin{align*}
    \norm{\hat{\mathbf{e}}_T - \mathbf{e}_T}_{\Sigma^{-1}(\mathbf{z}_k)}^2 - \norm{\mathbf{y}_{T}+\hat{\mathbf{e}}_T -\boldsymbol{\mu}(\mathbf{z_k}) }_{\Sigma^{-1}(\mathbf{z}_k)}^2\le \chi^2_{m}(\tau)
\end{align*}
\begin{align*}
    \norm{\hat{\mathbf{e}}_T - \mathbf{e}_T}_{\Sigma^{-1}(\mathbf{z}_k)}^2\le 2\chi^2_{m}(\tau)
\end{align*}
which implies that
\begin{align}\label{eqn:proof_thm2_2}
    \norm{\hat{\mathbf{e}}_T - \mathbf{e}_T}_2\le \sqrt{2\chi^2_{m}(\tau)\overline{\sigma}(\mathbf{z}_k)}.
\end{align}

Now, comparing \eqref{eqn:proof_thm2_1} and \eqref{eqn:proof_thm2_2} yields
\begin{align*}
    \norm{\hat{\mathbf{e}}_T - \mathbf{e}_T}_2&\le\min\left\{\begin{array}{c}K_3\sqrt{m-s}\norm{\mathbf{e}_{(T)}-\mathbf{e}_{(T)}[s]}_2,\\\\ \sqrt{2\chi^2_{m}(\tau)\overline{\sigma}(\mathbf{z}_k)}\end{array}\right\}\\\\
    &\le K_1\min\left\{K_2\norm{\mathbf{e}_{(T)}-\mathbf{e}_{(T)}[s]}_2,1\right\}\\\\
    &\le K_1\textsf{sat}_1\left(K_2\norm{\mathbf{e}_{(T)}-\mathbf{e}_{(T)}[s]}_2\right)
\end{align*}
\end{proof}

\section{Simulation Results}\label{s:experiments}
The attack-resilient observer proposed in this paper is evaluated using a numerical simulation of the IEEE 14-bus system shown in Fig. \ref{fig:IEEE14Bus}, it has $n_b = 14$ buses and $n_g = 5$ generators.
It is expected that each bus in the network has IIoT measurement devices able to provide active power injections and flow measurements.

\subsection{Model Description}
A small signal model is derived by linearizing the generator swing equations and power flow equations around the operating points under the assumption that:
\begin{itemize}
\item Voltage is tightly controlled at their nominal value;
\item Angular difference between each bus is small;
\item Conductance is negligible therefore the system is lossless.
\end{itemize}

\begin{figure}
    \centering{\includegraphics[width=3.3in]{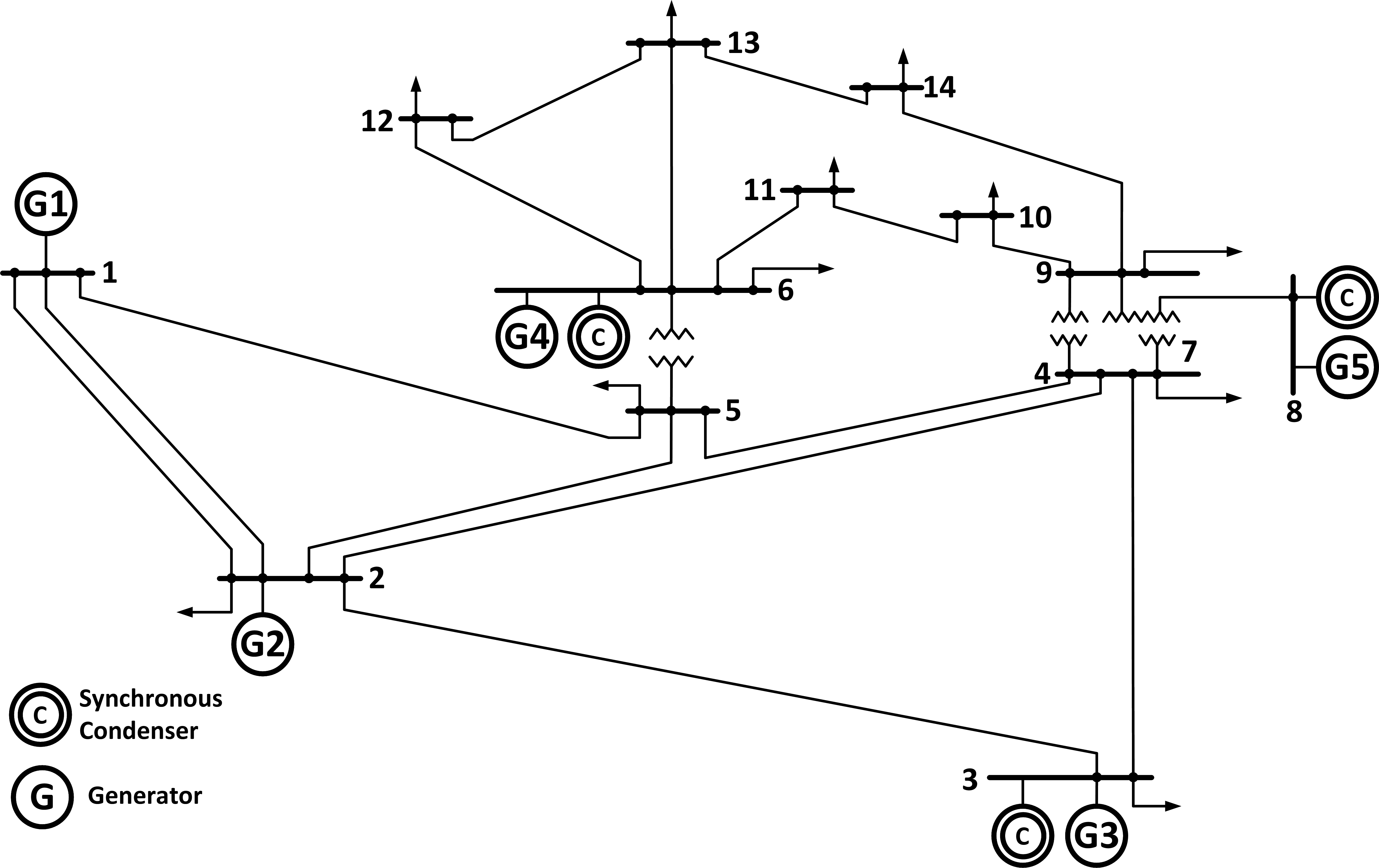}}
    \vspace{-2mm}
    \caption{IEEE 14-bus system.}
    \vspace{-2mm}
    \label{fig:IEEE14Bus}
\end{figure}

\noindent
Furthermore, the buses are ordered so that the first buses are generators, then the admittance-weighted \emph{Laplacian matrix} is expressed as
$L = \left[\begin{smallmatrix} L_{gg} & L_{lg}\\ L_{gl} & L_{ll} \end{smallmatrix}\right] \in \mathbb{R}^{N \times N}$, where $N = n_g+n_b$.
Thus, allowing the system to be described by the dynamic linearized swing equations and the algebraic DC power flow equations in the following manner:
\begin{equation}
    \begin{aligned}
    \begin{bmatrix}
    I & 0 & 0 \\
    0 & M & 0 \\
    0 & 0 & 0
    \end{bmatrix}
    \dot{x} &= -
    \begin{bmatrix}
    0 & -I & 0 \\
    L_{gg} & D_g & L_{lg} \\
    L_{gl} & 0 & L_{ll}
    \end{bmatrix}
    x +
    \begin{bmatrix}
    0 & 0 \\
    I & 0 \\
    0 & I
    \end{bmatrix}
    u
    \end{aligned}
\label{eqn:ieee14_system}
\end{equation}
The state variables, $x = [\delta^\top \ \omega^\top \ \theta^\top]^\top \in \mathbb{R}^{2n_g+n_b}$, consist of $\delta \in \mathbb{R}^{n_g}$ the generator rotor angle, $\omega \in \mathbb{R}^{n_g}$ the generator frequency and $\theta \in \mathbb{R}^{n_b}$ the voltage bus angles.
The control input $u = [P_g^\top \ P_d^\top ]^\top \in \mathbb{R}^{n_g+n_b}$ consists of $P_g \in \mathbb{R}^{n_g}$ the mechanical input power from each generator, controlled in a closed-loop manner with a PI regulating the generator frequency, and $P_d \in \mathbb{R}^{n_b}$ the active power demand at each bus.
Where $M$ is a diagonal matrix of inertial constants for each generator and $D_g$ is a diagonal matrix of damping coefficients.

For the system described in \eqref{eqn:ieee14_system}, the algebraic variable $\theta$ is eliminated to reduce the system dynamics to two state variables, $\tilde{x} = [\delta^\top \ \omega^\top]^\top \in \mathbb{R}^{2n_g}$, as follows:
\begin{equation}
    \begin{aligned}
        \begin{bmatrix}
            \dot{\delta}(t) \\
            \dot{\omega}(t) \\
        \end{bmatrix}
        = &
        \begin{bmatrix}
            0 & I \\
            -M^{-1}(L_{gg}-L_{gl} L_{ll}^{-1} L_{lg}) & -M^{-1} D_g\\
        \end{bmatrix}
        \tilde{x}
        \\
         &\hspace{1cm}+
        \begin{bmatrix}
            0 & 0 \\
            M^{-1} & -M^{-1}L_{gl}L_{ll}^{-1} \\
        \end{bmatrix}
        u,\\\\
        y(t) = &
        \begin{bmatrix}
            0 & I \\
            -P_{\text{node}} L_{ll}^{-1} L_{lg} & 0 \\
        \end{bmatrix}
        \tilde{x}
        + 
        \begin{bmatrix}
            0 & 0 \\
            -P_{\text{node}} L_{ll}^{-1} & 0 \\
        \end{bmatrix}
        u
    \end{aligned}
    \label{eqn:ieee14_system_obsv}
\end{equation}
Where $P_\text{node}$ is a function of the system incidence and susceptance matrices, obtained by linearizing the active power injections at the buses\cite{scholtz2004observer}. Consequently, the bus angles vector $\theta$(t) is given by:
\begin{equation*}
    \theta (t) = -L_{ll}^{-1}(L_{lg}\delta(t) - P_d).
\end{equation*}
The measurement channels $y(t) = [\omega^\top \ P_\text{net}^\top]^\top \in \mathbb{R}^{n_g+n_b}$ contain $\omega$ the generator frequency used in the PI feedback loop and $P_\text{net}$ the net power injected at each bus. 

\begin{figure}
    \centering{\includegraphics[scale = 0.65]{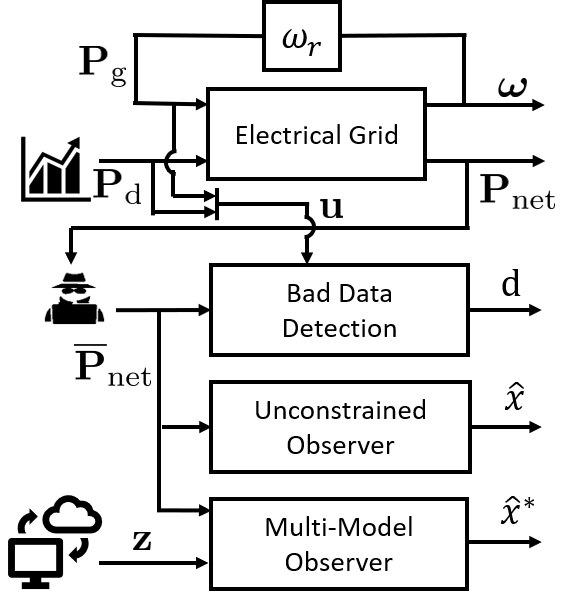}}
    \caption{Block diagram depiction of the simulation scenario.}
    \label{fig:SystemBlockDiagram}
\end{figure}
\subsection{Auxiliary Model}
In our earlier work \cite{anubi2019enhanced,anubi2019resilient}, we used actual data collected from the New York independent system operator (NYISO) to build gaussian process regression (GPR) models which map from market variables (namely, \texttt{locational bus marginal prices, marginal cost loses and marginal cost congestion}) to load data (namely, \texttt{active and reactive power}). The trained GPR model is executed to give the mean $\boldsymbol{\mu}(\mathbf{z})$ and the covariance $\Sigma(\mathbf{z})$ of the data-driven auxiliary model. For the scenario, presented here, we used the obtained covariance matrix from the earlier work to locate $\boldsymbol{\mu}(\mathbf{z}_k)$ within $3$ standard deviations of the true values in the simulation. This allows us to compensate for lack of data in this particular case by recreating a typical auxiliary model from known instances.

\subsection{Simulation Process}
The system in \eqref{eqn:ieee14_system_obsv} is discretized for the implementation of optimization problem in \eqref{eqn:enh_res_est_convex_1}. As depicted in the simulation scenario shown in Fig.~\ref{fig:SystemBlockDiagram}, $\overline{P}_\text{net}$ represents a compromised measurement. The goal of the observer design is to give the correct estimate of the rotor angle $\delta(t)$ under FDIA. We also implemented a residue-based bad data detection \cite{koglin1990bad} to monitor the integrity of the FDIA used in the simulation.

\subsection{Results}
We implemented three observer schemes; (i) a discretized Luenberger Observer, (ii) an Unconstrained $l_1$ minimization based Observer based on the results presented in the preliminary section, and (iii) a constrained Multi-model Observer implementing the optimization problem given in \eqref{eqn:enh_res_est_convex_1}, using the auxiliary model as described above and in Fig.~\ref{fig:SystemBlockDiagram}.

Fig.~\ref{fig:LO}, Fig.~\ref{fig:L1O} and Fig.~\ref{fig:MMO} shows the results of each observer when subjected to attacks on $30\%$ of the available measurement. The attack was triggered after $200$ samples into into the simulation. This is done to show the observers performances before and after attacks. 

Fig. \ref{fig:BDD} shows the integrity of the FDIA used in the simulation. As seen in the figure, the attacked measurement is able to pass a residue-based bad data detection (BDD) test. The threshold used in the simulation is $0.05$ and is assumed to be known to an attacker. Different threshold values were used and the FDIA passed the BDD test for all of them. We only show one plot here due to space limitation.

To further clarify the performances depicted in the figures above, we present in Table \ref{table:error_table} two metrics to give numerical comparison of the performance of each observer. The first metric is the root-mean-square (RMS) value, which quantifies the \emph{energy} of the errors between the actual value and estimated values. The second metric is the maximum absolute value of the error. This helps us capture the worst-case performance of each of the observers compared. It is seen also that, in both metrics considered, the proposed multi-model observer outperforms the other two.

\begin{figure}
    \centering{\includegraphics[width=3.36in]{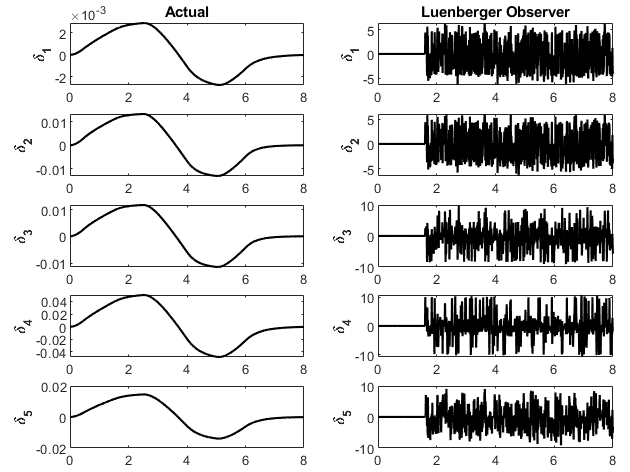}}
    \caption{A comparison of the actual $\delta(t)$ and the observed $\delta(t)$ by a discretized Luenberger Observer. It is seen that any attacks on the measurement channels leads to an inability to reconstruct the real states.}
    \label{fig:LO}
\end{figure}
\begin{figure}
    \centering{\includegraphics[width=3.36in]{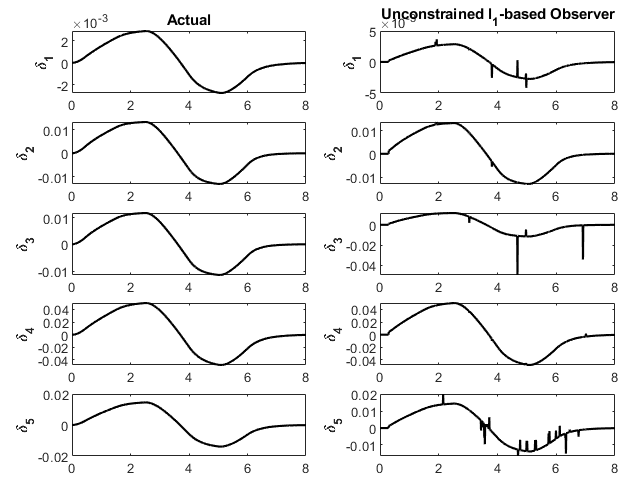}}
    \caption{A comparison of the actual $\delta(t)$ and the observed $\delta(t)$ by the Unconstrained ${l}_1$-minimization-based Observer. Although, this observer is able to reconstruct most of the signals, there are still outliers that could cause instability if used as a feedback to a controller.}
    \label{fig:L1O}
\end{figure}    
\begin{figure}
    \centering{\includegraphics[width=3.36in]{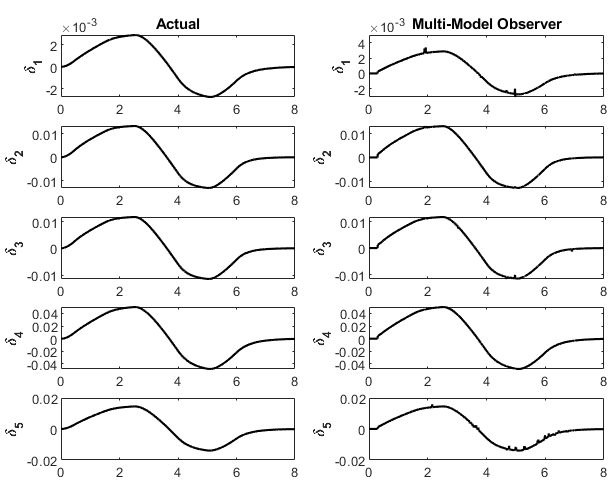}}
    \caption{A comparison of the actual $\delta(t)$ and the observed $\delta(t)$ by the proposed Multi-Model Observer. The proposed observer is able reconstruct the $\delta(t)$ to within a much more reasonable degree of accuracy. This is as a result of the additional information given through the constraint from the auxiliary model.}
    \label{fig:MMO}
\end{figure}
\begin{figure}
    \centering{\includegraphics[width=3.36in]{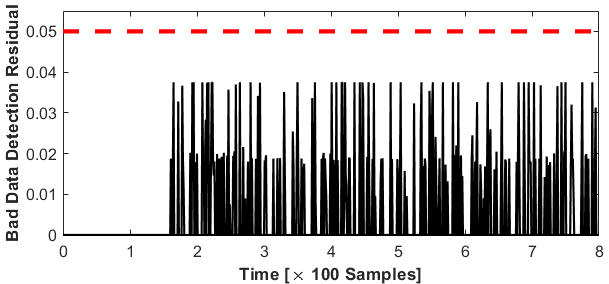}}
    \caption{A FDIA is undetectable by a Bad Data Detector that is set to a $5\%$ threshold}
    \label{fig:BDD}
\end{figure}

\begin{table}[H]
\centering
\caption{Error Metric Values}
\begin{tabular}{ |c|c|c|c|c|c|c| }
 \hline
 &\multicolumn{3}{|c|}{\textbf{RMS metric}}&\multicolumn{3}{|c|}{\textbf{Max. Abs. metric}}\\
 \hline
 &\textbf{LO}&\textbf{L1O} &\textbf{MMO}&\textbf{LO}&\textbf{L1O} &\textbf{MMO}\\
 \hline
 $\delta_1$&2.8801 & 0.0001 &	0.0001 & 6.4274 & 0.0028 &	0.0007\\
 $\delta_2$&2.7967 & 0.0002	& 0.0001 &	6.4437 & 0.0022 &	0.0013\\
 $\delta_3$&3.2746 & 0.0018	& 0.0001 &	9.7444 & 0.0387 &	0.0013\\
 $\delta_4$&3.4786 & 0.0004	& 0.0004 &	10.7019 & 0.0048 &	0.0042\\
 $\delta_5$&3.329 &	0.0011	& 0.0003 &	9.1387 & 0.0121 &	0.0024\\
 \hline
 \multicolumn{7}{|l|}{\textbf{LO:} Luenberger Observer, \textbf{L1O}: Unconstrained $\ell_1$-based Observer}\\
 \multicolumn{7}{|l|}{\textbf{MMO}: Proposed Multi-Model Observer}\\
 \hline
\end{tabular}
\label{table:error_table}
\end{table}

\section{Conclusions}\label{s:conclusions}
In this paper, a constrained optimization based resilient state observer is developed using $l_1$ minimization scheme.
The novelty of the algorithm lies in its ability to take into account the machine learning model as a constraint.
This constraint, the physics-based model and estimation theory is what makes this multi-model observer resilient.
The developed algorithm is evaluated through a numerical example of IEEE-14 bus system.
Under FDIA, state measurements differ from their true state.
By incorporating the resilient observer the FDIAs can be neutralized and true states can be retrieved with further accuracy.

Some of the problems open for future work include observing the behaviour of resilient observer as filter by cascading it in closed loop with the controller. Considering more complicated constraints for reconstruction optimization problem, rather than a simple quadratic constraint. We would also study the behavior of resilient observer under FDIA and model uncertainty, and extend our approach to other CPS.

\bibliographystyle{IEEEtran}
\bibliography{refs}

\begin{thebibliography}{10}
\providecommand{\url}[1]{#1}
\csname url@samestyle\endcsname
\providecommand{\newblock}{\relax}
\providecommand{\bibinfo}[2]{#2}
\providecommand{\BIBentrySTDinterwordspacing}{\spaceskip=0pt\relax}
\providecommand{\BIBentryALTinterwordstretchfactor}{4}
\providecommand{\BIBentryALTinterwordspacing}{\spaceskip=\fontdimen2\font plus
\BIBentryALTinterwordstretchfactor\fontdimen3\font minus
  \fontdimen4\font\relax}
\providecommand{\BIBforeignlanguage}[2]{{%
\expandafter\ifx\csname l@#1\endcsname\relax
\typeout{** WARNING: IEEEtran.bst: No hyphenation pattern has been}%
\typeout{** loaded for the language `#1'. Using the pattern for}%
\typeout{** the default language instead.}%
\else
\language=\csname l@#1\endcsname
\fi
#2}}
\providecommand{\BIBdecl}{\relax}
\BIBdecl

\bibitem{konstantinou2015cyber}
C.~Konstantinou \emph{et~al.}, ``Cyber-physical systems: A security
  perspective,'' in \emph{20th IEEE European Test Symposium (ETS)}.\hskip 1em
  plus 0.5em minus 0.4em\relax IEEE, 2015, pp. 1--8.

\bibitem{wu2018bad}
Y.~Wu \emph{et~al.}, ``Bad data detection using linear {WLS} and sampled values
  in digital substations,'' \emph{IEEE Transactions on Power Delivery},
  vol.~33, no.~1, pp. 150--157, 2018.

\bibitem{liu2011false}
Y.~Liu, P.~Ning, and M.~K. Reiter, ``False data injection attacks against state
  estimation in electric power grids,'' \emph{ACM Transactions on Information
  and System Security (TISSEC)}, vol.~14, no.~1, p.~13, 2011.

\bibitem{liang2017review}
G.~Liang \emph{et~al.}, ``A review of false data injection attacks against
  modern power systems,'' \emph{IEEE Transactions on Smart Grid}, vol.~8,
  no.~4, pp. 1630--1638, 2017.

\bibitem{deng2017false}
R.~Deng \emph{et~al.}, ``False data injection on state estimation in power
  systems -- attacks, impacts, and defense: A survey,'' \emph{IEEE Transactions
  on Industrial Informatics}, vol.~13, no.~2, pp. 411--423, 2017.

\bibitem{hao2015sparse}
J.~Hao \emph{et~al.}, ``Sparse malicious false data injection attacks and
  defense mechanisms in smart grids,'' \emph{IEEE Transactions on Industrial
  Informatics}, vol.~11, no.~5, pp. 1--12, 2015.

\bibitem{konstantinou2016case}
C.~Konstantinou and M.~Maniatakos, ``A case study on implementing false data
  injection attacks against nonlinear state estimation,'' in \emph{Proceedings
  of the 2nd ACM Workshop on Cyber-Physical Systems Security and
  Privacy}.\hskip 1em plus 0.5em minus 0.4em\relax ACM, 2016, pp. 81--92.

\bibitem{konstantinou2017gps}
C.~Konstantinou \emph{et~al.}, ``Gps spoofing effect on phase angle monitoring
  and control in a real-time digital simulator-based hardware-in-the-loop
  environment,'' \emph{IET Cyber-Physical Systems: Theory \& Applications},
  vol.~2, no.~4, pp. 180--187, 2017.

\bibitem{9087789}
A.~{Sayghe}, O.~M. {Anubi}, and C.~{Konstantinou}, ``Adversarial examples on
  power systems state estimation,'' in \emph{2020 IEEE Power Energy Society
  Innovative Smart Grid Technologies Conference (ISGT)}, 2020, pp. 1--5.

\bibitem{ashok2018online}
A.~Ashok, M.~Govindarasu, and V.~Ajjarapu, ``Online detection of stealthy false
  data injection attacks in power system state estimation,'' \emph{IEEE
  Transactions on Smart Grid}, vol.~9, no.~3, pp. 1636--1646, 2018.

\bibitem{anubi2019enhanced}
O.~M. {Anubi} and C.~{Konstantinou}, ``Enhanced resilient state estimation
  using data-driven auxiliary models,'' \emph{IEEE Transactions on Industrial
  Informatics}, vol.~16, no.~1, pp. 639--647, Jan 2020.

\bibitem{fawzi2014secure}
H.~Fawzi, P.~Tabuada, and S.~Diggavi, ``Secure estimation and control for
  cyber-physical systems under adversarial attacks,'' \emph{IEEE Transactions
  on Automatic Control}, vol.~59, no.~6, pp. 1454--1467, 2014.

\bibitem{hu2016secure}
Q.~Hu \emph{et~al.}, ``Secure state estimation for nonlinear power systems
  under cyber attacks,'' \emph{arXiv preprint arXiv:1603.06894}, 2016.

\bibitem{Fiore2017Secure}
G.~Fiore \emph{et~al.}, ``Secure state estimation for cyber physical systems
  with sparse malicious packet drops,'' in \emph{American Control Conference
  (ACC), 2017}.\hskip 1em plus 0.5em minus 0.4em\relax IEEE, 2017, pp.
  1898--1903.

\bibitem{8894484}
C.~{Konstantinou} and M.~{Maniatakos}, ``A data-based detection method against
  false data injection attacks,'' \emph{IEEE Design Test}, pp. 1--1, 2019.

\bibitem{Mishra2015Secure}
S.~Mishra \emph{et~al.}, ``Secure state estimation: Optimal guarantees against
  sensor attacks in the presence of noise,'' in \emph{Information Theory
  (ISIT), 2015 IEEE International Symposium on}.\hskip 1em plus 0.5em minus
  0.4em\relax IEEE, 2015, pp. 2929--2933.

\bibitem{Liu2017Secure}
X.~Liu, Y.~Mo, and E.~Garone, ``Secure dynamic state estimation by decomposing
  {K}alman filter,'' \emph{IFAC-PapersOnLine}, vol.~50, no.~1, pp. 7351--7356,
  2017.

\bibitem{mestha2017cyber}
L.~K. Mestha, O.~M. Anubi, and M.~Abbaszadeh, ``Cyber-attack detection and
  accommodation algorithm for energy delivery systems,'' in \emph{Control
  Technology and Applications (CCTA), 2017 IEEE Conference on}.\hskip 1em plus
  0.5em minus 0.4em\relax IEEE, 2017, pp. 1326--1331.

\bibitem{anubi2018robust}
O.~M. Anubi, L.~Mestha, and H.~Achanta, ``Robust resilient signal
  reconstruction under adversarial attacks,'' \emph{arXiv preprint
  arXiv:1807.08004}, 2018.

\bibitem{candes2005decoding}
E.~J. Candes and T.~Tao, ``Decoding by linear programming,'' \emph{IEEE
  transactions on information theory}, vol.~51, no.~12, pp. 4203--4215, 2005.

\bibitem{hayden2016sparse}
D.~Hayden \emph{et~al.}, ``Sparse network identifiability via compressed
  sensing,'' \emph{Automatica}, vol.~68, pp. 9--17, 2016.

\bibitem{pajic2017design}
M.~Pajic \emph{et~al.}, ``Design and implementation of attack-resilient
  cyberphysical systems: With a focus on attack-resilient state estimators,''
  \emph{IEEE Control Systems}, vol.~37, no.~2, pp. 66--81, 2017.

\bibitem{donoho2003optimally}
D.~L. Donoho and M.~Elad, ``Optimally sparse representation in general
  (nonorthogonal) dictionaries via l1 minimization,'' \emph{Proceedings of the
  National Academy of Sciences}, vol. 100, no.~5, pp. 2197--2202, 2003.

\bibitem{elad2002generalized}
M.~Elad and A.~M. Bruckstein, ``A generalized uncertainty principle and sparse
  representation in pairs of bases,'' \emph{IEEE Transactions on Information
  Theory}, vol.~48, no.~9, pp. 2558--2567, 2002.

\bibitem{gribonval2003sparse}
R.~Gribonval and M.~Nielsen, ``Sparse representations in unions of bases,''
  \emph{IEEE transactions on Information theory}, vol.~49, no.~12, pp.
  3320--3325, 2003.

\bibitem{tropp2004greed}
J.~A. Tropp, ``Greed is good: Algorithmic results for sparse approximation,''
  \emph{IEEE Transactions on Information theory}, vol.~50, no.~10, pp.
  2231--2242, 2004.

\bibitem{cai2013sparse}
T.~T. Cai and A.~Zhang, ``Sparse representation of a polytope and recovery of
  sparse signals and low-rank matrices,'' \emph{IEEE transactions on
  information theory}, vol.~60, no.~1, pp. 122--132, 2013.

\bibitem{pasqualetti2013attack}
F.~Pasqualetti, F.~D{\"o}rfler, and F.~Bullo, ``Attack detection and
  identification in cyber-physical systems,'' \emph{IEEE Transactions on
  Automatic Control}, vol.~58, no.~11, pp. 2715--2729, 2013.

\bibitem{anubi2019resilient}
O.~M. Anubi, C.~Konstantinou, and R.~Roberts, ``Resilient optimal estimation
  using measurement prior,'' \emph{arXiv preprint arXiv:1907.13102}, 2019.

\bibitem{scholtz2004observer}
E.~Scholtz, ``Observer-based monitors and distributed wave controllers for
  electromechanical disturbances in power systems,'' Ph.D. dissertation,
  Massachusetts Institute of Technology, 2004.

\bibitem{koglin1990bad}
H.-J. Koglin, T.~Neisius, G.~Bei$\beta$ler, and K.~Schmitt, ``Bad data
  detection and identification,'' \emph{International Journal of Electrical
  Power \& Energy Systems}, vol.~12, no.~2, pp. 94--103, 1990.

\end{thebibliography}

\end{document}